\documentclass{amsart}
\usepackage{lmodern,graphicx}
\newtheorem{theorem}{\bf Theorem}[section]

\newcommand{\RR}{\mathbb{R}}
\newcommand{\abs}[1]{\left\vert #1\right\vert}

\title{The dynamic pressure in deep-water extreme Stokes waves}

\author{
Tony Lyons}

\address{Department of Computing and Mathematics,\\ Waterford Institute of Technology,\\ Cork Road,  Waterford,\\ Ireland}



\email{tlyons@wit.ie}
\usepackage[pdfdisplaydoctitle=true,colorlinks=true,urlcolor=blue,citecolor=blue,linkcolor=blue,bookmarksnumbered=true]{hyperref}

\begin{document}
\maketitle

\begin{abstract}
 In this paper we consider the dynamic pressure in a deep-water extreme Stokes wave. While the presence of stagnation points introduces a number of mathematical complications, maximum principles are applied to analyse the dynamic pressure in the fluid body by means of an excision process. It is shown that the dynamic pressure attains its maximum value beneath the wave crest and its minimum beneath the wave trough, while it decreases in moving away from the crest line along any streamline.
\end{abstract}

\section{Introduction}
The extreme Stokes wave, or the wave of greatest height in Stokes' own terminology, is a steady, periodic, two-dimensional flow, propagating in water of infinite depth or over a flat bed of finite depth. In general, Stokes waves are periodic, steady waves, propagating with fixed velocity. The surface profiles of these waves feature a single crest and trough per period, while the surface height above the mean level is decreasing between a crest and subsequent trough. In the case of irrotational Stokes waves, it is always the case that this surface profile is described by the graph of a function, in that there is no overhang in the profile cf. \cite{Spi1970, Var2007}.  

In contrast to smooth Stokes waves, the wave of greatest height displays some unusual features such as the presence of stagnation points at the wave crests, which arise as peaks with an included angle of $120^\circ$.  Such flows were conjectured by Stokes \cite{Sto1880}, as a limit of regular surface gravity waves, whose existence he deduced from formal asymptotic series applied to the governing equations modelling inviscid, incompressible flows. Stokes conjecture was
proved almost 100 years later \cite{AFT1982}. Let us point out that while numerical simulations in \cite{LGL2002} seem to indicate the existence of two-dimensional irrotational symmetric periodic gravity waves having a stagnation point of the flow field inside the fluid domain and for which the horizontal water velocities near the crest exceed the wave speed, subsequent rigorous analytical considerations show that these waves are a numerical abberation, see \cite{Con2004}. Physically, extreme Stokes waves are unusual in that they are not the most energetic or the most impulsive Stokes wave, nor do they   move with the largest velocity, cf. \cite{SF1982}.

In this work we shall investigate the behaviour of the dynamic pressure induced by a wave of greatest height propagating in a fluid body of infinite depth. The behaviour of the dynamic pressure in smooth Stokes waves has been recently investigated in \cite{Con2016} by means of maximum principles, and the reader is also referred to \cite{SW2016, Var2009} for further recent analytic results relating to smooth and extreme Stokes waves. In contrast to smooth Stokes waves, the presence of stagnation points along the free-surface of the wave of greatest height requires a more intricate application of maximum principles so successfully applied in the aforementioned work. This lack of regularity means the free surface does not satisfy all the necessary criteria for the application of maximum principles in the fluid body. However by means of an excision process it is possible to circumvent these complications, while a limiting argument allows us to extend the results to the entire fluid body, thereby providing the qualitative features of the dynamic pressure in extreme Stokes waves.

The hydrodynamic pressure distribution is composed of contributions from the atmospheric pressure, the hydrostatic pressure and the dynamic pressure generated by the passing wave, and it is this latter contribution which we shall investigate in this work. In the absence of the dynamic pressure contribution, the hydrostatic and atmospheric contributions counteract the fluid inertia, which imposes a state of equilibrium on the fluid body. The hydrodynamic pressure distribution is known to play a vital role in the dynamics of a fluid body. For instance the pressure gradient is known to govern the velocity field and thereby the fluid particle trajectories within the fluid body. In the case of large amplitude surface gravity waves, full nonlinearity of the model must be accounted for, which leads to the surprising result of particle drift along the streamlines of these flows, and we refer to the following \cite{Con2006, CE2007, Hen2006, Hen2008, Con2012,Lyo2014,OVDH2012} for recent results concerning particle drift in Stokes waves, while the work \cite{Qir2017} investigates fluid particle paths subject to Coriolis forces in periodic equatorial waves. Moreover from a theoretical and experimental perspective the pressure induced by a gravity wave is of value when constructing the surface profile of the wave, and the reader is referred to \cite{CC2013,CEH2011,Con2012a,ES2008,Hen2011,Hen2013, Kog2015, HT2017, Bas2017} for further discussion of the recovery of periodic and aperiodic steady surface profiles from pressure data. In addition, in maritime settings where large amplitude waves predominate, an understanding of the hydrodynamic pressure is crucial to the design of cost effective offshore structures \cite{DKD1992}.

\section{The free boundary problem}

In the following we shall consider an extreme Stokes wave of fixed wavelength $\lambda$, propagating with a fixed velocity $c$ with respect to the physical frame with coordinates $(X,Y,Z)$. In practice, such steady waves might arise when wind generated waves enter a remote body of still water, that is to say still with respect to the physical frame coordinates.  The horizontal coordinate $X$ corresponds to the direction of wave propagation, and without loss of generality we may align our axes such that $c>0$. The vertical coordinate is denoted by $Z$ while the time variable is denoted by $t$. Due to the underlying two-dimensionality of the fluid motion, the remaining $Y$-coordinate plays no role in our further considerations. 
\begin{figure}[h!]
\centering
\includegraphics[]{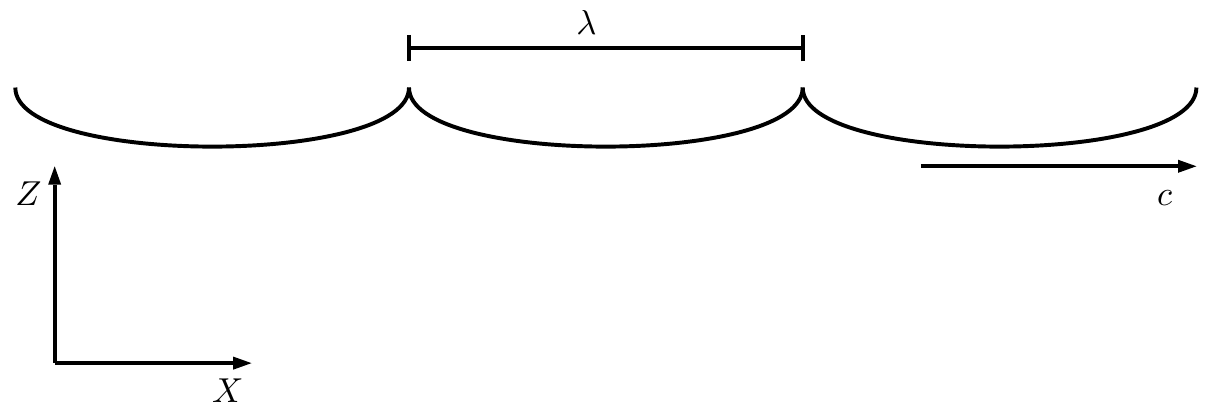}
\caption{A cross-section of a deep-water extreme Stokes' wave, of wavelength $\lambda$, propagating with fixed velocity $c$ in the physical frame with coordinates $(X,Z)$.}
\end{figure}
Being a periodic, steady, two-dimensional flow, without loss of generality we may restrict our considerations to the fluid throughout one period of the flow, whose closure is given by
\begin{equation}\label{eq2.1}
\bar{\Omega }=\left\{(x,y) \in \RR^2: -\frac{\lambda }{2}\leq X\leq \frac{\lambda }{2},\ -\infty<Z\leq \eta (X,t)\right\},
\end{equation}
where $\eta(X,t)$ is the surface height above the mean surface level $Z=0$. The fluid motion is described by a time-dependent two-component vector field $\mathbf{u}=(u,w)$, where $u(X,Z,t)$ denotes the horizontal velocity of the fluid element and $w(X,Z,t)$ denotes its vertical velocity. The horizontal velocity $u(X,Z,t)$ and pressure $P(X,Z,t)$ are periodic in $X$ with period $\lambda$ and are symmetric about any crest line, while the vertical velocity $w(X,Z,t)$ is $\lambda-$periodic in $X$ and anti-symmetric about the crest line. Meanwhile, the free surface profiel $\eta(X,t)$ is periodic in $X$ with period $\lambda$ and is also symmetric about any crest line, and possesses one peak and one trough per period, while it is strictly decreasing between any crest and subsequent trough, cf. \cite{CE2004a, CE2004b, OS2001}.

The governing equations describing the dynamic behaviour of the fluid in the physical frame are given by Euler's equation:
\begin{equation}\label{eq2.2}
\left.
\begin{aligned}
 &u_{t}+uu_X+wu_Z=-\frac{1}{\rho }P_X\\
 &w_{t}+uw_X+ww_Z=-\frac{1}{\rho }P_Z-g
\end{aligned}
\right\}\quad\text{for }(X,Z) \in \Omega,
\end{equation}
where $\rho$ denotes fluid density, while
\begin{equation}\label{eq2.3}
 P=P_{A}-\rho gZ +\mathcal{P}(X,Z,t)
\end{equation}
is the pressure exerted on a fluid element at $(X,Z,t)$. There are two distinct contributions to this pressure: $P_A-\rho gz$ the hydrostatic pressure being a combination of atmospheric pressure ($P_A$) and static fluid pressure ($-\rho gZ$), with the remaining contribution being the dynamic pressure $\mathcal{P}(X,Z,t)$ induced by the passing wave.

While regular waves (that is, steady, periodic waves moving with constant speed) are known to arise in the presence of underlying currents, cf. \cite{CE2011b,Con2011,CE2004b,CE2004a,CV2011,Con2013,CSV2016, TK1997}, in this work we shall restrict our attention to irrotational flows.
The density of water is dependent on the temperature, salinity and pressure exerted on that water body. Variations in water temperature observed in the ocean typically account for variations in observed fluid density of approximately $0.5\%$. Meanwhile, the variation in salinity accounts for a change in fluid density up to $0.2\%$, while density variations of up to $0.2\%$ per kilometre are observed due to pressure variations with fluid depth \cite{Con2011}. Consequently, in the presence of wind generated ocean waves we may reasonably neglect any variation in fluid density, see \cite{Kin1965, Lig1978} for further discussion. Combined with the divergence theorem, this uniform density hypothesis leads to the equation of mass conservation given by
\begin{equation}\label{eq2.4}
 u_X+w_Z=0,\quad\text{for }(X,Z) \in \Omega.
\end{equation}
Moreover, since we shall consider wind generated waves entering a remote area of still water, we do not need to account for the effects of wave-current interactions, which in turn ensures we may treat the flow generated by the passing wave as irrotational. This ensures that the flow satisfies the irrotationality condition
\begin{equation}\label{eq2.5}
 u_{Z}-w_{X}=0,\quad\text{for }(X,Z) \in \Omega
\end{equation}
that is to say, the velocity field $\mathbf{u}$ must be curl-free.   Combined, equations \eqref{eq2.4}--\eqref{eq2.5} also ensure that both $u$ and $w$ are harmonic in the interior $\Omega$.

The governing equations of the fluid motion are also complimented by associated  boundary conditions on the free surface $Z=\eta (X,t)$ and at great depth $Z\rightarrow-\infty$:
\begin{equation}\label{eq2.6}
\begin{aligned}
 &\left.
\begin{aligned}
  & P=P_{A}\\
  & w=\eta_t+u\eta_X\
\end{aligned}\right\}\quad\text{on }Z=\eta \\
 & (u,w)\to(0,0)\quad\text{ as }Z \rightarrow -\infty.
\end{aligned}
\end{equation}
The dynamic boundary condition reflects the decoupling of free surface motion from atmospheric motion. In addition, the kinematic boundary condition imposed on the free surface ensures that fluid particles on the free surface remains there, while the second kinematic boundary condition ensures the fluid is essentially motionless at great depth. The reader is referred to \cite{Con2011} for a comprehensive derivation of the system \eqref{eq2.2}--\eqref{eq2.6} for irrotational two-dimensional steady flows.

In the case of extreme Stokes waves we make the following general observations concerning the velocity field $\mathbf{u}$, the pressure $P$ and the free surface $\eta$, all of which are observed in smooth Stokes waves \cite{CS2010} and in the limit of extreme waves. Each of the variables $u$, $w$ and $P$ depend on $X-ct$ and $Z$, while $\eta$ depends only on $X-ct$. Each of the variables is periodic in $X$, with a common periodicity $\lambda$. The free surface $\eta$, the pressure $P$ and the horizontal velocity $u$ are all symmetric about any crest line while the vertical velocity $w$ is anti-symmetric about any crest line, where a crest line indicates the vertical line issuing from a wave crest into the fluid body.

In contrast to smooth Stokes waves, we no longer assume the solution of the free boundary problem is smooth, since the free surface is no longer differentiable at the wave-crest and the velocity field exhibits stagnation points at the wave crest also. The presence of these stagnation points in the wave profile introduces a number of mathematical complications, which we must circumvent by excising the stagnation point and using a limiting argument. This excision method has been successfully utilised to analyse the hydrodynamic pressure in extreme Stokes waves as well as the particle trajectories along streamlines of the flow, in both the deep-water and shallow-water context, cf. \cite{Con2012, Lyo2014, Lyo2016a, Lyo2016b}.

\section{The fixed boundary problem}
\subsection{The moving frame system}
Since the extreme Stokes wave is a steady waveform moving with fixed phase speed $c$, it is natural to investigate its behaviour in the moving frame coordinates
\begin{equation}\label{eq3.1}
 x=X-ct\qquad z=Z.
\end{equation}
In the moving frame, the solution of the free boundary problem, given by the functions $\{\mathbf{u}(x,z), \eta (x,z), P(x,z)\}$ only depend
on time implicitly via the $x$-coordinate. Under this change of variables the Euler equation becomes
\begin{equation}\label{eq3.2}
\left.
\begin{aligned}
 (u-c)u_x+wu_z&=-\frac{1}{\rho }P_x\\
 (u-c)w_x+ww_z&=-\frac{1}{\rho }P_z-g
\end{aligned}
\right\}\quad\text{ for }(x,z) \in \Omega,
\end{equation}
where $\Omega$ denotes the interior of the fluid domain:
\[\Omega:=\left\{(x,z) \in \RR^2:-\frac{\lambda }{2}<x<\frac{\lambda }{2}, -\infty<z<\eta (x)\right\}.\]
The incompressibility and irrotationality of the flow ensure
\begin{equation}\label{eq3.3}
\left.
\begin{aligned}
  u_x+w_z=0\\
  u_z-w_x=0
\end{aligned}
\right\}\quad\text{ for }(x,z) \in \Omega,
\end{equation}
while the  boundary conditions in the moving frame become
\begin{equation}\label{eq3.4}
\begin{aligned}
&\left.
\begin{aligned}
P&=P_{A}\\
w&=(u-c)\eta_x
\end{aligned}
\right\}\quad\text{ on }z=\eta (x)\\
&\ \  w=0\quad\text{ on }z=0,
\end{aligned}
\end{equation}
which completes the free boundary problem.

Owing to the  periodicity of the extreme Stokes wave solution $\left\{u,w,P,\eta \right\}$, we may restrict our considerations to the portion of the fluid body given by
\begin{equation}\label{eq3.5}
\begin{aligned}
 \Omega_{+}&=\left\{(x,z) \in \RR^2:0<x<\frac{\lambda }{2}, -\infty<z<\eta (x)\right\}\\
 \Omega_{-}&=\left\{(x,z) \in \RR^2:-\frac{\lambda }{2}<x<0, -\infty<z<\eta (x)\right\}.
\end{aligned}
\end{equation}
These regions share a common boundary given by the crest line
\begin{equation}\label{eq3.6}
\{(x,z) \in \RR^2:x=0,-\infty<z\leq\eta (0)\},
\end{equation}
which we abbreviate as $\{x=0\}$. In addition these internal domains are bounded above by the free surface segments
\begin{equation}
\begin{aligned}\label{eq3.7}
S_{+}&=\left\{(x,z) \in \RR^2:0<x<\frac{\lambda }{2}, z=\eta (x)\right\}\\
S_{-}&=\left\{(x,z) \in \RR^2:-\frac{\lambda }{2}<x<0, z=\eta (x)\right\}
\end{aligned}
\end{equation}
while also being bounded laterally by the trough lines
\begin{equation}\label{eq3.8}
\left\{(x,z) \in \RR^2:x=\pm\frac{\lambda }{2},-\infty<z\leq\eta \left(\pm\frac{\lambda }{2}\right)\right\},
\end{equation}
which we abbreviate as $\left\{x=\pm\frac{\lambda }{2}\right\}$.

\subsection{The hodograph transform}
\subsubsection{The stream function}
The incompressibility condition \eqref{eq3.3} ensures the existence of a stream function $\psi (x,y)$ for the flow, defined by
\begin{equation}\label{eq3.9}
 \psi_{x}=-w\qquad \psi_{z}=u-c,
\end{equation}
which combined with irrotationality of the flow immediately implies $\psi$ must be harmonic in $\Omega$.
Integrating, we find
\begin{equation}\label{eq3.10}
 \psi (x,z)=\int_{z_0}^{z}\left(u(x_0,s)-c\right)ds-\int_{x_0}^{x}w(l,z)dl.
\end{equation}
The $\lambda$-periodicity of the velocity field $\mathbf{u}(x,z)$ along the $x$-axis ensures $\psi (x,z)$ shares the same periodicity in this direction. Moreover, by an appropriate choice of integration constant, we may always ensure $\psi (x,\eta )=0$, while $\psi_{z}=u(x,z)-c<0$ everywhere except at $(0,\eta (0))$ ensures $\psi \rightarrow\infty$ as $z\to-\infty$, uniformly in $x$.

Noting this, the free-boundary problem may then be reformulated as a linear elliptic problem for $\psi$ coupled with a nonlinear boundary condition at the free boundary, given by
\begin{equation}\label{eq3.11}
\begin{cases}
\Delta \psi = 0\quad \text{ for }(x,z) \in \Omega\\
\psi=0\quad\text{ on }z=\eta (x)\\
\psi \rightarrow\infty\quad\text{ as }z\to-\infty\\
\frac{1}{2g}\abs{\nabla\psi }^2+z=Q\quad\text{ on }z=\eta (x).
\end{cases}
\end{equation}
The nonlinear boundary condition is a reformulation of the Bernoulli condition (restricted to the free surface), given by
\begin{equation}\label{eq3.12}
 \frac{u^2+w^2}{2g}+z+\frac{P-P_A}{\rho g}=Q\quad\text{ for }(x,z) \in \bar{\Omega },
\end{equation}
which is interpreted as a form of energy conservation throughout the closure of the fluid domain $\bar{\Omega }$. The constant $Q$ is the hydraulic head, and the reader is referred to \cite{Con2011,Con2012,Con2016} for further discussion of its physical significance.

\subsubsection{The velocity potential}
Irrotationality of the flow ensures the velocity field $\mathbf{u}$ may be written in terms of a velocity potential $\varphi$, according to
\begin{equation}\label{eq3.14}
    \varphi_{x}=u-c\quad\varphi_{z}=w,
\end{equation}
and again the fact that the flow is both irrotational and incompressible also ensures the velocity potential is harmonic in the interior of the fluid domain.
Integrating, we find that the velocity potential may be written according to
\begin{equation}\label{eq3.15}
\begin{aligned}
\varphi (x,z)=\int_{x_0}^{x}\left(u(l,z)-c\right)dl+\int_{z_0}^{z}w(x_0,s)ds
\end{aligned}
\end{equation}
where the integration constant is chosen such that $\phi (0,z)=0$. The $\lambda-$periodicity of the vector field also ensures
\begin{equation}\label{eq3.16}
 \varphi (x+\lambda,z)=\varphi (x,z)-c\lambda.
\end{equation}

\subsubsection{The conformal mapping}
Introducing the variables
\begin{equation}\label{eq3.17}
  q=-\phi (x,z)\qquad p=-\psi (x,z),
\end{equation}
the hodograph transform given by
\begin{equation}\label{eq3.18}
\begin{aligned}
 &\mathcal{H}:\bar{\Omega } \rightarrow \bar{\hat{\Omega }}\\
 &\mathcal{H}:(x,z) \mapsto (q,p)
\end{aligned}
\end{equation}
is a conformal mapping (in the interior) of the fluid domain onto a fixed boundary domain.
\begin{figure}\label{eq3.19}
\centering
\includegraphics[width=\textwidth]{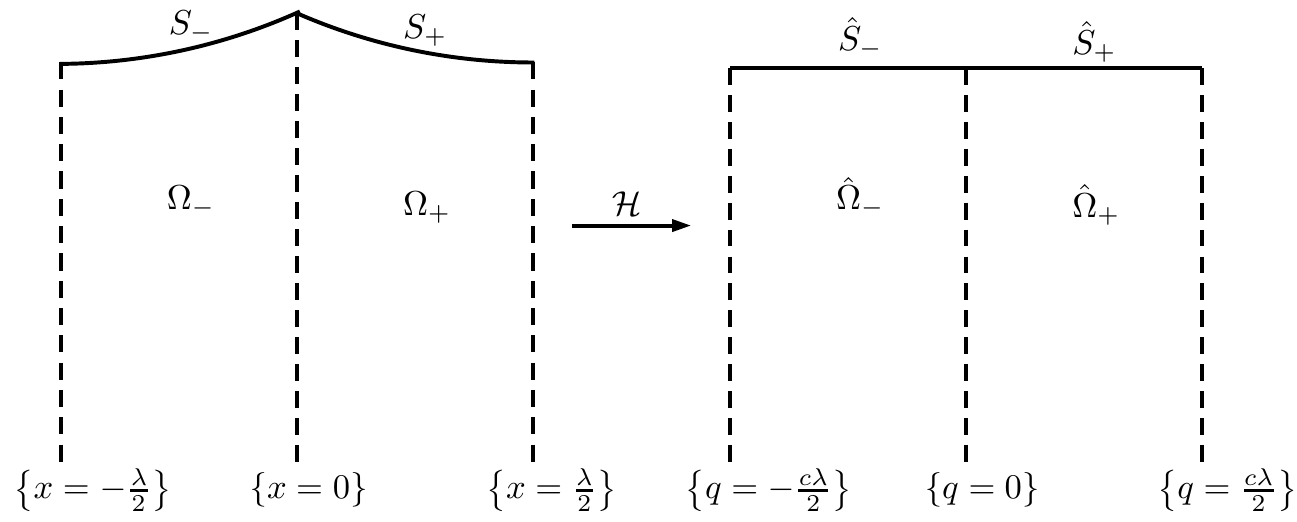}
\caption{The hodograph transform, mapping the free boundary domain $\Omega$ to the fixed boundary domain $\hat{\Omega }$.}
\end{figure}
Under this hodograph transform we  find that the interiors $\Omega_+$ and $\Omega_-$ maps to the interiors
\begin{equation}\label{eq3.20}
\begin{aligned}
\hat{\Omega }_+&=\left\{(q,p) \in \RR^2:0<q<\frac{c\lambda }{2}, -\infty<p<0\right\}\\
\hat{\Omega }_-&=\left\{(q,p) \in \RR^2:-\frac{c\lambda }{2}<q<0, -\infty<p<0\right\}
\end{aligned}
\end{equation}
while the free surface segments $S_+$ and $S_-$ map to the horizontal segments
\begin{equation}\label{eq3.21}
\begin{aligned}
\hat{S}_+&=\left\{(q,p):0<q<\frac{c\lambda }{2}, p=0\right\}\\
\hat{S}_-&=\left\{(q,p):-\frac{c\lambda }{2}<q<0, p=0\right\}.
\end{aligned}
\end{equation}
The crest-line $\{x=0\}$ and trough-lines $\left\{x=\pm \frac{\lambda }{2}\right\}$ map to the vertical line segments
\[\{(q,p):q=0,-\infty<p\leq0\}\quad\text{and}\quad\left\{\left(q,p\right):q=\pm\frac{c\lambda }{2},-\infty<p\leq 0\right\},\]
which we abbreviate as $\{q=0\}$ and $\left\{q=\pm\frac{c\lambda }{2}\right\}$ respectively.

Introducing the mapping
\begin{equation}\label{eq3.22}
\begin{aligned}
 &h:\bar{\hat{\Omega }}\to\bar{\Omega }\\
 &h:(q,p) \mapsto y,
\end{aligned}
\end{equation}
we note that $h$ is manifestly harmonic in the fluid domain $\Omega,$ and since the hodograph transform is conformal in the interior, it follows that $h$ is harmonic in the conformal domain $\hat{\Omega }$, cf. \cite{Fra2000}. The boundary value problem given by equation \eqref{eq3.11} may be written as a fixed boundary problem in the conformal domain according to
\begin{equation}
\begin{cases}
\Delta_{q,p}h=0 \quad\text{ for }(q,p) \in \hat{\Omega }_{-} \cup \hat{\Omega }_+\\
2(Q-P_{A}-gh)(h_{q}^2+h_{p}^2)\quad\text{ for }p=0\\
\nabla h\to\left(0,\frac{1}{c}\right)\text{ uniformly in $q \in \left(-\frac{c\lambda }{2},\frac{c\lambda }{2}\right)$ as $p\to-\infty$.}
\end{cases}
\end{equation}
This reformulation of the free boundary problem is valid in the interior $\hat{\Omega }$ and by Carath\'{e}odery's theorem cf. \cite{Pom1992} the problem has a continuous extension to the boundary $\partial\hat{\Omega }$, see \cite{Con2012} for further discussion. Nevertheless, the lack of regularity at the wave crest $(q,p)=(0,0)$ means the function $h$ is not differentiable nor is the hodograph transform invertible there.
This is further illustrated by the following transformations between the bases $\{\partial_x,\partial_z\}$ and $\{\partial_q,\partial_p\}$:
\begin{equation}\label{eq3.23}
\begin{aligned}
&\left(
\begin{matrix}
  \partial_x\\
  \partial_z
\end{matrix}
\right) = \left(
	  \begin{matrix}
	    -(u-c)&w\\
	    -w&-(u-c)
	  \end{matrix}
	  \right)\left(
		  \begin{matrix}
		   \partial_q\\
		   \partial_p
		  \end{matrix}
		 \right)\\
&\left(
\begin{matrix}
  \partial_q\\
  \partial_p
\end{matrix}
\right) = \frac{1}{(u-c)^2+w^2}\left(
	  \begin{matrix}
	    -(u-c)&-w\\
	    w&-(u-c)
	  \end{matrix}
	  \right)\left(
		  \begin{matrix}
		   \partial_x\\
		   \partial_z
		  \end{matrix}
		 \right)
\end{aligned}
\end{equation}
where we observe that $\{\partial_q,\partial_p\}$ are not defined at the stagnation point.

\section{The dynamic pressure}\label{s4}
\subsection{The uniform limiting process}
The following nonlinear integral equation due, to Nekrasov \cite{Nek1920}, may be derived from the Bernoulli condition for a solution of the free boundary problem in equation \eqref{eq3.11} (see \cite{Mil2011}, pp 409--413 for a full derivation)
\begin{equation}\label{eq4.1.1}
\theta (s)=\frac{1}{3\pi }\int_{-\pi }^{\pi }\sum_{n=1}^{\infty}\frac{\sin(nr)\sin(ns)}{n}\frac{\sin(\theta (t))}{\frac{1}{\mu }+\int_{0}^{t}\sin(\theta (\sigma ))d\sigma }dt,\sigma \in\mathbb{C},
\end{equation}
and was used in \cite{Tol1978} to prove the existence of the wave of greatest height.  The function $\theta(s)$ is interpreted as the local slope of the free surface solution of the system \eqref{eq3.11}, namely
\begin{equation}\label{eq4.1.2}
\theta (s)=\arctan(\eta^{\prime}(x)),\quad s\in[-\pi,\pi],
\end{equation}
while the parameter $s$ is interpreted as follows: Under the action of a second hodograph transform, the point $(x,y) \in \bar{\Omega }$ is mapped to a point in the punctured unit-disc $\sigma=\rho e^{is} \in \mathbb{D}$, where $\rho=1$ corresponds to the free surface, while $\rho=0$ corresponds to the point $y=-\infty$. The parameter $\mu$ is given by
\begin{equation}\label{eq4.1.3}
\mu=\frac{3g\Lambda c}{2\pi U_0},\quad U_0=\abs{\nabla\psi (0,\eta (0))}.
\end{equation}
Conversely, given a solution to  equation \eqref{eq4.1.1} for $\mu>3$, then it can be shown to correspond to a deep-water solution of \eqref{eq3.11} for periodic, irrotational, inviscid flows.

In particular, the case $\mu \to\infty$ corresponds to $U_0\to0$, that is to say, a solution in which a stagnation point occurs at the wave-crest. Given an unbounded sequence $\left\{\mu_n\right\}$ with corresponding solutions  $\left\{\theta_n(s)\right\}$, it may be shown that  $\theta_n(s)$ converge weakly to some $\theta(s) \neq0$, with respect to the $L_{2}([0,2\pi])$ norm. Moreover, by means of the dominated convergence theorem, it may be shown that $\theta_n(s)$ converges strongly to $\theta (s)$ with respect to the same norm. As such, the existence of a solution for $\mu \to\infty$ or $U_0\to0$ has been shown to exist, cf. \cite{Tol1978}. Finally, it may be shown that the strong limit $\theta_n(s) \rightarrow \theta(s)$ satisfies the following:
\begin{enumerate}
 \item The Fourier series of $\theta (s)$ is given by\\
 \[\theta (s)=-\frac{1}{3\pi }\sum_{n=1}^{\infty}\left(\int_{-\pi }^{\pi }\cos(nt)\left[\ln\int_{0}^{t}\sin\theta (\omega )d\omega \right]dt\right)\sin(ks)\]
 \item This Fourier series is uniformly convergent for $s \in (0,\pi]$, and as such $\theta (s)$ is continuous on $[-\pi,0) \cup (0,\pi]$
 \item $\theta (s)$ is discontinuous at $s=0$.
\end{enumerate}
In this manner the wave of greatest height is realised as a uniform limit of regular Stokes, and in particular this allows us to deduce several features of the extreme wave from the corresponding behaviour of smooth surface gravity waves. Extending these methods, the work \cite{AFT1982} confirmed the surface profile is analytic at every point away from the wave-crest, which was also shown to have an included angle $\frac{2\pi }{3}$, that is to say $\displaystyle{\lim_{s\to0^+}}\theta (s)=\frac{\pi }{3}$. The remaining aspect of Stokes conjecture, namely the convexity of the surface profile $(x,\eta (x))$ between successive wave crests, was proven in \cite{PT2004}, and the reader is referred to \cite{BT2003,Tol1996} for further discussion relating to the existence and regularity of Stokes waves of extreme form.

\subsection{Main result}

Equation \eqref{eq2.3} written with respect to the moving-frame coordinates yields the following form for the dynamic pressure
\begin{equation}\label{eq4.1}
    \mathcal{P}(x,z)= P(x,z)-(P_A-\rho gz),
\end{equation}
the latter terms on the right hand side above corresponding to the hydrostatic pressure at fixed depth $z$. Recalling equation \eqref{eq3.12}, Bernoulli's condition along the free surface combined with the dynamic boundary condition $P(0,\eta (0))=P_A$, ensures that
\begin{equation}\label{eq4.3}
 \rho g\eta (0)+P_A=Q.
\end{equation}
Hence, the dynamic pressure may be reformulated as
\begin{equation}\label{eq4.4}
 \mathcal{P}(x,z)=\rho g\eta (0)-\frac{\rho }{2}\left((u(x,z)-c)^2+w(x,z)^2\right).
\end{equation}
Using this form of the dynamic pressure we now prove the main result of this paper:

\begin{theorem}\label{thm1}
In an extreme Stokes wave, the dynamic pressure attains its maximum value beneath the wave crest while its minimum value occurs beneath the wave trough.
\end{theorem}
\begin{proof}
Since $w$ is anti-symmetric in the case of case of regular almost extreme Stokes waves, cf. \cite{CS2010}, it follows that in the uniform limit of the extreme Stokes wave this property of $w$ is preserved. Thus, for the extreme wave, we find that $w=0$ along the crest line  $\{x=0\}$, while periodicity in the $x$-variable also ensures $w=0$ along the trough line $\left\{x=\frac{\lambda}{2}\right\}$. In the conformal domain $\hat{\Omega }^{+}$ this corresponds to $w=0$ along $\{q=0\}$ and $\left\{q=\frac{c\lambda}{2}\right\}$.

Since the boundary $\partial\Omega_{\varepsilon}^{+}$ is not regular at the wave crest, nor is the velocity field $\mathbf{u}$ differentiable there, this  means the requisite conditions for the application of maximum principles are not in place, see \cite{GT2001} for further discussion. However, since the velocity field and free surface are real-analytic at every point away from the wave crest, we may circumvent this issue by applying maximum principles in the excised conformal domain  $\Omega_{\varepsilon}^{+}$,  cf. Figure \ref{fig3}. This is the domain obtained from $\hat{\Omega}^+$ by cutting a quarter disc of radius $\varepsilon>0$ centered at the wave-crest $(q,p)=(0,0)$, and whose boundary $\partial\hat{\Omega}_{\varepsilon}^{+}$ is regular.

Thus in the excised conformal domain $w$ is harmonic in the interior $\hat{\Omega }_{\varepsilon }^{+}$ and continuous along the boundary $\partial{\hat{\Omega}}_{\varepsilon}^+$ (indeed $w$ is analytic along the boundary), while it attains its minimum value along the lateral edges. Hence, the strong maximum principle ensures $w_{q}(0,p)>0$ for $p \in (-\infty,-\varepsilon )$ and $w_{q}\left(\frac{c\lambda }{2},p\right)<0$ for $p \in (-\infty,0).$ Since $\varepsilon>0$ was chosen arbitrarily, we deduce that $w_q(0,p)=0$ for $p\in(-\infty,0)$. Consequently, since $w$ vanishes along the lateral edges of the fluid domain $\Omega_+$ and the conformal domain $\hat{\Omega }_+$, the irrotationality condition along with the linear system \eqref{eq3.23} ensure
\begin{equation}\label{eq4.5}
\begin{cases}
  u_{z}<0\text{ along } \{(x,z) \in \RR^2:x=0,-\infty<z<\eta (0)\},\\
  u_{z}>0\text{ along } \left\{(x,z) \in \RR^2:x=\frac{\lambda }{2},-\infty<z\leq\eta \left(\frac{\lambda }{2}\right)\right\},
\end{cases}
\end{equation}
having used $u-c<0$ for every point away from the wave crest. Differentiating equation \eqref{eq4.4} with respect to $z$, and using $w=0$ along the lateral edges and $u-c<0$ away
from the wave crest, it follows that
\begin{equation}\label{eq4.6}
\begin{aligned}
  \mathcal{P}_{z}&<0\text{ along } \{(x,z) \in \RR^2:x=0,-\infty<z<\eta (0)\},\\
  \mathcal{P}_{z}&>0\text{ along } \left\{(x,z) \in \RR^2:x=\frac{\lambda }{2},-\infty<z\leq\eta \left(\frac{\lambda }{2}\right)\right\}.
\end{aligned}
\end{equation}
\begin{figure}
\centering
\includegraphics{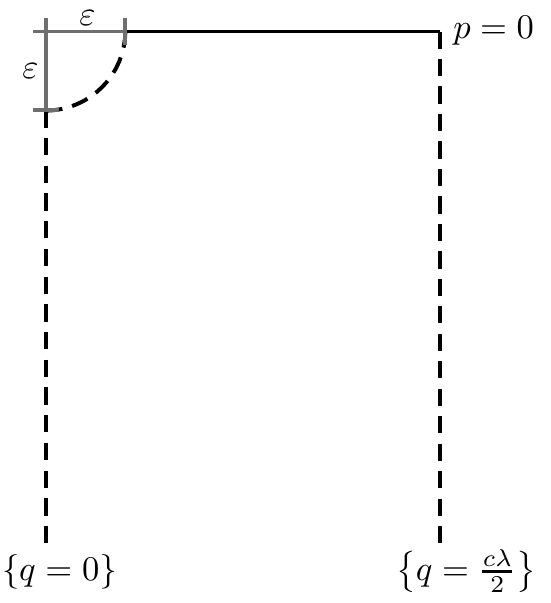}
\caption{The excised conformal domain $\hat{\Omega}_{\varepsilon}^{+}$.}\label{fig3}
\end{figure}
Given any streamline in the interior
$\left\{(x,z(x)): x \in \left(0,\frac{\lambda }{2}\right)\right\}\subset\Omega_+,$
differentiating the dynamic pressure along such a streamline we find
\begin{equation}\label{eq4.7}
 \mathcal{P}_{x}(x,z(x))=\frac{(u(x,z(x))-c)^2+w(x,z(x))^2}{c-u(x,z(x))}u_{x}(x,z(x)).
\end{equation}
It is known that $u_{x}$ is strictly decreasing along any streamline when moving between any crest line and subsequent trough line (see \cite{Con2012, Lyo2014} for further discussion) and since $c-u>0$ away from the wave-crest, it follows that
\begin{equation}\label{eq4.8}
  \mathcal{P}_x(x,z(x))<0 \quad \text{for }x \in (0,\pi ).
\end{equation}
That is to say, the dynamic pressure is strictly decreasing along any streamline as we move between a crest line and subsequent trough line.

Along the free boundary $(x,\eta(x))$ we have
\begin{equation}\label{eq4.9}
  P_{x}(x,\eta (x))\leq 0 \text{ for }x \in \left[0,\frac{\lambda }{2}\right],
\end{equation}
with the pressure gradient being identically zero only at the wave crest and wave trough, cf. \cite{Lyo2016a,Lyo2016b}. In relation to the dynamic pressure, equation \eqref{eq4.1} then ensures
\begin{equation}\label{eq4.10}
  \mathcal{P}_{x}(x,\eta (x))=P_{x}(x,\eta (x))+\rho g\eta^{\prime}(x).
\end{equation}
Since the surface profile $(x,\eta (x))$ is convex for $x\in[0,\lambda]$ we have $\eta^{\prime}(x)<0$ for $x \in (0,\pi )$. Moreover, considering only the right hand derivative of $\eta$ at $x=0$ we have $\displaystyle{\lim_{x\to0^{+}}}\eta^{\prime}(x)=-\frac{1}{\sqrt{3}}$ at the wave crest, while $\eta^{\prime}\left(\frac{\lambda }{2}\right)=0$ at the wave trough. Hence $\mathcal{P}(x,\eta (x))$ is decreasing along the free surface between the wave crest and the wave trough.

\begin{figure}[h!]
\centering
\includegraphics{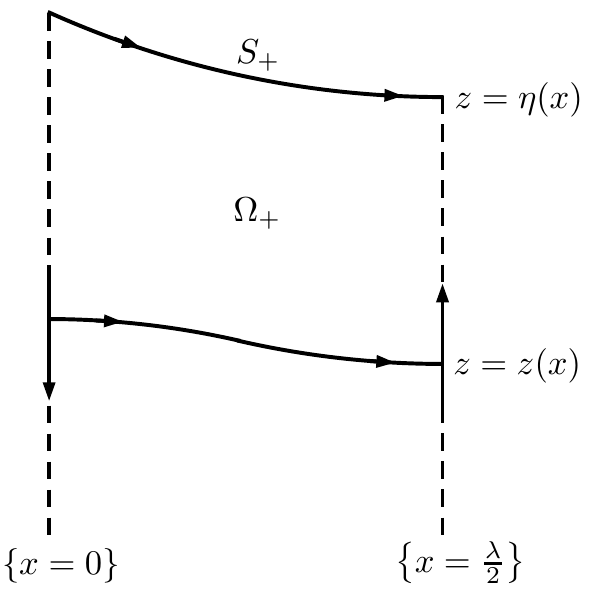}
\caption{The dynamic pressure gradient in a deep-water extreme Stokes wave. The direction of the arrows indicate the direction in which the dynamic pressure is decreasing.}\label{fig4}
\end{figure}

Since the interior $\Omega^+$ maps to the interior of $\hat{\Omega }^+$ under the hodograph transform, if we can show that the dynamic pressure $\mathcal{P}$ attains no extreme values in the interior of the conformal domain, then it follows that the dynamic pressure does not possess any extrema within the interior of the fluid domain. Within the interior of the conformal domain the dynamic pressure is given by
\begin{equation}
  \mathcal{P}(q,p)=\rho gh(0,0)-\frac{\rho }{2}\left((u(q,p)-c)^2+w(q,p)^2\right).
\end{equation}
Since the hodograph transform is conformal and the velocity field is harmonic in the interior of the fluid domain, it follows the velocity field is also harmonic in the interior of the conformal domain, see \cite{Fra2000}. Combined with the incompressibility condition and the irrotationality condition, it follows that
\begin{equation}\label{eq4.11}
\left.
\begin{aligned}
&\Delta \mathcal{P}+b_{1}(q,p)\mathcal{P}_{q}+b_{2}(q,p)\mathcal{P}_p=0\\
&b_{1}(q,p)=\frac{\mathcal{P}_{q}}{\rho \left((u(q,p)-c)^2+w(q,p)^2\right)}\\
&b_{2}(q,p)=\frac{\mathcal{P}_{p}}{\rho \left((u(q,p)-c)^2+w(q,p)^2\right)}.
\end{aligned}
\right\}\quad \text{for }(q,p) \in \hat{\Omega }^{+}
\end{equation}
Crucially, since $(u(q,p)-c)^2+w(q,p)^2>0$ in the interior of this excised conformal domain, this ensures the coefficients $b_1(q,p)$ and $b_2(q,p)$ remain bounded  in the interior of $\hat{\Omega }_{\varepsilon }^{+}$. Combined with the continuity of $\mathcal{P}$ along the boundary $\partial\hat{\Omega }_{\varepsilon }^{+}$, we may now apply maximum principles to the linear elliptic problem in equation \eqref{eq4.11}.
Indeed, Hopf's maximum principle now ensures that unless $\mathcal{P}$ is constant throughout the excised conformal domain $\hat{\Omega }_{\varepsilon }^{+}$, it cannot attain extreme values in the interior of this domain, cf. \cite{Con2011} Chapter 3. Since the parameter $\varepsilon>0$ is arbitrary, it follows that the dynamic pressure never attains its extreme values in the interior of the conformal domain $\hat{\Omega }^+$, and so the dynamic pressure $\mathcal{P}$ cannot attain an extreme value in the interior of the fluid domain. Combined with the behaviour of this dynamic pressure on the boundary $\partial\Omega^+$ elucidated above, we ensure Theorem \ref{thm1} is valid. In Figure \ref{fig4} the qualitative features of the dynamic pressure gradient are illustrated. This conforms with the behaviour of the dynamic pressure in smooth Stokes waves propagating over a flat bed of finite depth, cf. \cite{Con2016}.
\end{proof}

\section*{Acknowledgements}
The author is grateful to the referees for several helpful comments.

\end{document}